\documentclass[11pt,reqno]{amsart}

\usepackage{amsmath, euscript, amssymb, amsfonts}
\usepackage[matrix,arrow,curve]{xy}

\usepackage{mathrsfs}

\theoremstyle{plain}
\newtheorem{theorem}{Theorem}
\newtheorem{lemma}{Lemma}
\newtheorem{corollary}{Corollary}
\newtheorem{proposition}{Proposition}

\theoremstyle{definition}

\newtheorem*{acknowledgements}{Acknowledgements}

\theoremstyle{remark}
\newtheorem{remark}{Remark}

\newcommand{\coloneq}{\mathrel{\mathop:}=}

\renewcommand{\Im}{\mathop{\mathrm{Im}}\nolimits}

\newcommand{\twast}{\mathbin{\scriptstyle{\circledast}}}

\newcommand{\oR}{{\mathbb R}}
\newcommand{\oC}{{\mathbb C}}

\newcommand{\defeq}{\overset{\text{def}}{=}}

\newcommand{\labelT}[1]{\label{T#1}}
\newcommand{\labelP}[1]{\label{P#1}}
\newcommand{\labelL}[1]{\label{L#1}}

\newcommand{\labelC}[1]{\label{C#1}}

\newcommand{\refS}[1]{Section~\ref{S#1}}
\newcommand{\refT}[1]{Theorem~\ref{T#1}}
\newcommand{\refP}[1]{Proposition~\ref{P#1}}
\newcommand{\refL}[1]{Lemma~\ref{L#1}}

\begin{document}

\title{Twisted convolution  and Moyal  star product of
generalized functions}

\author{M.~A.~Soloviev}
\address{Lebedev
Physics Institute, Russian Academy of Sciences, Leninsky Prospect
53, Moscow 119991, Russia} \email{soloviev@lpi.ru}

\thanks{}


\begin{abstract}
We consider nuclear function spaces on which the Weyl-Heisenberg
group acts continuously and study the basic properties of the
twisted convolution product of  the functions with  the dual space
elements. The final theorem characterizes the corresponding
algebra of convolution multipliers  and shows that it contains all
sufficiently rapidly decreasing functionals in the dual space.
Consequently, we obtain a general description of the Moyal
multiplier algebra of the Fourier-transformed space. The results
extend the Weyl symbol calculus beyond the traditional framework
of tempered distributions.
\end{abstract}

\hfill  FIAN-TD/2012-15



\vspace{2cm}

 \maketitle
\section{Introduction}\label{S1}
The twisted convolution product of functions $g_1(s)$ and $g_2(s)$
on a linear symplectic space is a noncommutative deformation of
the ordinary convolution and is defined by the formula
 \begin{equation}
(g_1\twast_\theta g_2)(s)=\int   g_1(t)
g_2(s-t)\,e^{\tfrac{i}{2}\theta(s, t)}dt, \label{1.1}
\end{equation}
where $\theta$ is the bilinear skew-symmetric form  specifying the
symplectic structure.\footnote{The twisted convolution operation
is often denoted by  $*_\theta$, but  we use the symbol
$\twast_\theta$ here to avoid confusion with the star product
$\star_\theta$.} The Fourier transform converts the twisted
convolution into the Weyl-Groenewold-Moyal star product
$\star_\theta$ which gives the composition rule for the Weyl
symbols of quantum mechanical operators and plays a key role in
the Weyl quantization, see~\cite{BS}, \cite{ZFC}. It is worth
noting that the composition rule for  phase-space functions, which
corresponds to the composition of operators on a Hilbert space,
was originally written by von Neumann~\cite{N}  just in terms of
the twisted convolution.  It is customary  to define the star
multiplication and twisted convolution first for smooth and
rapidly decreasing functions in the Schwartz space $S$, which
forms an associative topological algebra under either of the two
operations. But in practice, depending on the problem under study,
we must consider an extension of these operations to one or
another subspace of the dual space $S'$ of tempered distributions.
Antonets proposed  a maximal extension by duality that consisted
in constructing the multiplier algebra of the algebra
$(S,\star_\theta)$~\cite{A1}-\cite{A3} or, equivalently, of the
algebra $(S,\twast_\theta)$. This extension was later studied in
many papers and most thoroughly in~\cite{Mail}-\cite{E} (a
detailed review and references can be found in~\cite{G}).

Field theory models  on noncommutative spaces based on using the
Moyal star product  have been intensively investigated  in the
last 15 years, (see, e.g.,~\cite{Sz} for an introduction to this
topic). The  interest in noncommutative spaces was stimulated by
the in-depth analysis of the quantum limitations on the accuracy
of localization of space-time events in quantum theory including
gravity~\cite{DFR1}, \cite{DFR2}  and by the study of the low
energy limit of string theory~\cite{SW}. There is
reason~\cite{Alv}-\cite{S06} to believe that the framework of
tempered distributions is too narrow for a consistent formulation
of the general principles of noncommutative quantum field theory.
The Moyal product is nonlocal and its expansion in powers of the
noncommutativity parameter $\theta$ converges only on analytic
test functions whose Fourier transforms decrease at infinity
faster than the Gaussian exponential function~\cite{I}, \cite{Ch}.
An analysis of microcausality violations in the simplest
noncommutative models~\cite{Green}-\cite{S08}  indicates a
possible connection between noncommutative field theory and the
previously considered nonlocal theories, which treat quantum
fields as operator-valued generalized functions defined on a
suitable space of analytic test functions instead of the Schwartz
space. The problem of appropriately  generalizing the Weyl symbol
calculus arises.

In~\cite{I}, we established a condition under which a nuclear test
function space $E(\oR^d)$ with the structure of a topological
algebra under ordinary convolution is also an algebra under the
twisted convolution  and its Fourier-conjugate space is  hence an
algebra under the Moyal product. This condition can be written as
 \begin{equation}
e^{-\tfrac{i}{2}\theta}\in M(E\mathbin{\widetilde{\otimes}}E),
\label{1.2}
\end{equation}
where $E\mathbin{\widetilde{\otimes}}E$ is the completed
projective tensor product identifiable with  $E(\oR^{2d})$ by the
kernel theorem and $M(E\mathbin{\widetilde{\otimes}}E)$ is the
space of its multipliers with respect to  ordinary pointwise
multiplication. In the case of Gel'fand-Shilov spaces
$S^\alpha_\beta$~\cite{GS2} considered in~\cite{I},
condition~\eqref{1.2} leads to the restriction $\alpha\ge\beta$ on
the specifying indices.   Palamodov precisely described the set of
pointwise multipliers for the spaces $S^\alpha_\beta$~\cite{P},
and their corresponding algebras of Moyal multipliers were
constructed in~\cite{S11}.

Here, we  show that the problem under discussion can be solved in
a general form and a Moyal multiplier algebra can be constructed
for any complete, nuclear, barrelled\footnote{Practically all
spaces used in the theory of generalized functions have these
properties, see~\cite{Sch} for their definition and role in
functional analysis.} space on which the Weyl-Heisenberg group
acts continuously. We also  describe the elements of these
algebras using only well-known facts from the theory of tensor
products of nuclear spaces~\cite{Sch}, which allows  avoiding
complicated analytic estimates. The basic observation is that
condition~\eqref{1.2} implies that
 \begin{equation}
v\twast_\theta g\in M(E)\quad\text{and}\quad g\twast_\theta v\in
M(E)\quad \text{for all $g\in E$, $v\in E'$,} \label{1.3}
\end{equation}
where $M(E)$ is the space of pointwise multipliers for $E$. This
allows characterizing  those elements of the dual space $E'$ that
are multipliers of the algebra $(E,\twast_\theta)$ sufficiently
exactly.

The paper is organized as follows.  In Sec.~\ref{S2}, we consider
function spaces endowed with a continuous linear action of the
Weyl-Heisenberg group and define the twisted convolution of
elements of the space with  elements of its dual. We also present
the noncommutative analogues of some properties of the usual
convolution. In Sec.~\ref{S3} we introduce the basic notion of a
twisted convolution  multiplier. Section~\ref{S4} contains the
main theorems;  there, we show that the implication~\eqref{1.2}
$\Rightarrow$~\eqref{1.3} holds under natural assumptions on the
spaces under consideration, and using this implication, we obtain
a characterization of the algebra of twisted convolution
multipliers. Section~\ref{S5} is devoted to the most interesting
case of spaces invariant under the Fourier transform. In
Sec.~\ref{S6} we  illustrate the general construction with
concrete examples. The appendix contains the proof of a useful
simple criterion for the continuity of the action of topological
groups on spaces of the  considered type.

\section{The Weyl-Heisenberg group and the twisted \\convolution}
\label{S2}
 Let  $E$ be a locally convex space of complex-valued
functions on $\oR^d$ and let  $\theta$  be a bilinear symplectic
form, i.e., a nondegenerate skew-symmetric inner product on
$\oR^d$. We assume that the twisted shift operator
\begin{equation}
\tau_{\theta,s}\colon g(t)\rightarrow e^{\tfrac{i}{2}\theta
(s,t)}g(t-s),\qquad g\in E, \label{2.1}
\end{equation}
is defined and continuous on $E$  for each $s\in \oR^d$. The
operators $\tau_{\theta,s}$ realize a projective representation of
the translation group to which corresponds a linear representation
of its  central extension generated by the initial symplectic
structure, i.e of the Weyl-Heisenberg group consisting of elements
of the form $a=(\alpha,s)$, where $\alpha\in\oR$ and $s\in\oR^d$,
with the multiplication law
\begin{equation}
a_1a_2=\left(\alpha_1+\alpha_2+\dfrac{1}{2}\theta(s_1,
s_2),s_1+s_2\right).
 \notag
\end{equation}
We also assume that $E$ is invariant under the complex conjugation
$g\to g^*$.  Then, along with the representation $a\to
e^{i\alpha}\tau_{\theta,s}$, the conjugate representation $a\to
e^{-i\alpha}\bar\tau_{\theta,s}$ is realized on $E$ with
\begin{equation}
\bar\tau_{\theta,s}\colon g(t)\rightarrow
e^{-\tfrac{i}{2}\theta(s, t)}g(t-s). \label{2.2}
\end{equation}
If $E$ is also invariant under the coordinate reflection $g(t)\to
\check{g}(t)\coloneq g(-t)$, then the functions
\begin{equation}
(v\twast_\theta g)(s)\coloneq \left\langle
v,e^{\tfrac{i}{2}\theta(s,\cdot)}g(s-\cdot)\right\rangle,\quad
(g\twast_\theta v)(s)\coloneq \left\langle
v,e^{-\tfrac{i}{2}\theta(s,\cdot)}g(s-\cdot)\right\rangle,
 \label{2.3}
\end{equation}
are well defined for any $g\in E$ and  any $v\in E'$. We call them
the twisted convolution products of the function $g$ with the
functional $v$.

\begin{remark}\label{R1} The symplectic form entering in
the definition of the Weyl-Heisenberg group is usually assumed to
be standard and determined by the matrix $J=\left(\begin{matrix}
0&I_n\\-I_n&0\end{matrix}\right)$, where $2n=d$. But for any other
form  $\theta$, the group is obviously the same up to an
isomorphism because every symplectic space has a symplectic
basis~\cite{A}.  In the Wigner-Weyl representation, or phase-space
formulation of quantum mechanics, the twisted convolution is
determined by the matrix  $\hbar J$, where the Planck constant
$\hbar$ plays the role of a noncommutative deformation parameter.
\end{remark}

Although operators~\eqref{2.1} and~\eqref{2.2} and also
functions~\eqref{2.3} carry the subscript $\theta$, we  omit it in
what follows if it cannot cause confusion. We  also assume that
$E$ is a dense vector subspace of the Schwartz space $S$, although
definition~\eqref{2.3} is meaningful  even in a more general
situation, for example, when $E$ is continuously and densely
embedded into $L^2$ and  there is  hence a canonical embedding
$E\to E'$ and formulas~\eqref{2.3} extend the initial twisted
convolution operation on the elements of  $E$. As usual, we let
$C(\oR^d)$ denote the space of all continuous functions on $\oR^d$
and give it the topology of uniform convergence on compact sets.

\begin{proposition}\labelP{1}
Let $E$ be a locally convex space of complex-valued functions on
$\oR^d$. If the reflection transformation $g\to\check g$ and
operators~\eqref{2.1} and \eqref{2.2} act continuously on $E$,
then functions~\eqref{2.3} belong to $C(\oR^d)$. If, in addition,
$E$ is barrelled, then the maps $(g, v)\to v\twast g$ and $(g,
v)\to g\twast v$ from $E\times E'$ to $C(\oR^d)$ are separately
continuous.
\end{proposition}

\begin{proof} Formulas~\eqref{2.3} can be rewritten as
\begin{equation}
(v\twast )(s)=\langle v,\tau_s \check g\rangle,\qquad (g\twast
v)(s)=\langle v,\bar\tau_s \check g\rangle.
 \notag
\end{equation}
The continuity of the action of  $\tau_s$ on $E$ means that  the
map $\oR^d\times E\to E\colon (s,g)\to\tau_sg$ is continuous.
Hence, functions~\eqref{2.3} are obviously continuous in $s$ under
our assumptions.  If $g$ is fixed and $s$ ranges over a compact
set $K$ in $\oR^d$, then $\tau_s\check g$ ranges over a compact
set $Q$ in $E$. Therefore, the map $v\to v\twast g$ is continuous
in the strong topology of $E'$ by its definition as the topology
of uniform convergence on  bounded sets. Analogously, the map
$v\to g\twast v$ is  continuous. Now let $v$ be fixed. The set of
operators $\tau_s$, $s\in K$, is pointwise bounded; hence, if $E$
is barrelled, then this set is equicontinuous by the general
Banach-Steinhaus theorem (Sec.~III.4.2 in \cite{Sch}). Therefore,
for any neighborhood  $U$ of the origin in $E$, there exists a
neighborhood  $V$ such that $\tau_s\check g\in U$ for all $g\in V$
and for any $s\in K$. Taking $U$ to be a neighborhood on which $v$
is bounded by a number $\epsilon>0$, we conclude that the map
$g\to v\twast g$ is continuous. It can be  similarly verified that
the map $g\to g\twast v$ is continuous. The proposition is proved.
\end{proof}

\begin{remark}\label{R2} Our main theorems concern complete nuclear
spaces. Such spaces are barrelled if and only if they are
reflexive and then they are Montel spaces~\cite{Sch}. In the
appendix we prove~\refL{1}, which gives easily verifiable
sufficient conditions for the continuity of the action of
topological transformation groups  on  Montel function spaces.
\end{remark}

The operators $\tau_s$ are automorphisms of  $C(\oR^d)$ and it is
easily seen that the linear map  $L_v\colon g\to v\twast g$
commutes with  them. Analogously, the map $R_v\colon g\to g\twast
v$ commutes with  all operators $\bar\tau_s$.

\begin{proposition}\labelP{2}
Let $E$ satisfy the conditions of~\refP{1} and let $A$ be a
continuous linear map from $E$ to $C(\oR^d)$. If $A$ commutes with
all the operators $\tau_s$ (with all $\bar\tau_s$), then there
exists a unique functional $v\in E'$ such that $Ag=v\twast g$
$(Ag=g\twast v)$ for all $g\in E$.
\end{proposition}

\begin{proof}
The argument is analogous to that used for Theorem~4.2.1
in~\cite{H1}, where the ordinary convolution and the space
$E=C^\infty_0(\oR^d)$ were considered. We consider the case where
$A$ commutes with all  $\tau_s$. If such a functional $v$ exists,
then  $\langle v, \check g\rangle=(Ag)(0)$ for all $g\in E$ and is
hence  unique. We let $v$ denote the linear form $g\to (A\check
g)(0)$. Clearly,  $v\in E'$ because the operator $A$ and the
transformation $g\to\check g$ are continuous. Furthermore,
$(v\twast g)(0)=\langle v, \check g\rangle=(Ag)(0)$. Using the
commutativity of $A$ with $\tau_s$, we obtain
\begin{equation}
(Ag)(-s)=(\tau_s Ag)(0)=(A\tau_s g)(0)=(v\twast\tau_s
g)(0)=(v\twast g)(-s)\quad\text{for all $s\in\oR^d$}.
 \notag
\end{equation}
Therefore, $Ag=v\twast g$ for all $g\in E$. The second case can be
treated in the same way. The proposition is proved.
\end{proof}

\section{Twisted convolution multipliers}
\label{S3}
 In what follows, we assume that  $E$ consists of
continuous functions and the embedding $E\to C(\oR^d)$ is
continuous.  The conditions in Proposition~\ref{P1} are also
assumed to be satisfied. We define the spaces ${\mathcal
C}_{\theta,L}(E)$ and ${\mathcal C}_{\theta,R}(E)$ of left and
right $\twast$-multipliers for $E$ as consisting of all
functionals $v\in E'$ such that respectively $v\twast_\theta g\in
E$ and  $g\twast_\theta v\in E$ for each $g\in E$ and, in
addition,  the maps $g\to v\twast_\theta g$ and $g\to
g\twast_\theta v$ are continuous.

If  the closed graph theorem is applicable to  $E$, then the
continuity requirement is  satisfied automatically here and can be
omitted from the definition because these maps have closed graphs.
Indeed, if $g_\gamma\to g$ and $v\twast g_\gamma\to h$ in $E$,
then $v\twast g_\gamma\to h$ in $C(\oR^d)$ and  $v\twast g= h$
by~\refP{1}. It follows from~\refP{2} that there is a one-to-one
correspondence between the elements of ${\mathcal C}_L(E)$ and the
continuous linear maps  $E\to E$ commuting with all $\tau_s$.
Letting $\mathcal  L(E)$ denote the algebra of continuous linear
operators on $E$ equipped with the topology of uniform convergence
on  bounded sets,  we conclude that ${\mathcal C}_L(E)$ can be
identified with its closed subalgebra. If $v_1,v_2\in {\mathcal
C}_L(E)$, then we let $v_1\twast v_2$ denote the element of
${\mathcal C}_L(E)$ corresponding to the composition of operators
$g\to v_1\twast g$ and $g\to v_2\twast g$ by Proposition~\ref{P2}.
Explicitly, we have
\begin{equation}
\langle v_1\twast v_2,\check g\rangle= (v_1\twast(v_2\twast
g))(0),\quad g\in E.
 \label{3.1}
\end{equation}
The natural topology of  ${\mathcal C}_L(E)$ is the topology
induced by that of ${\mathcal L}(E)$. Analogously, ${\mathcal
C}_R(E)$ is identified with the subalgebra of operators that
commute with all $\bar\tau_s$. For the elements of ${\mathcal
C}_R(E)$, the twisted convolution product is defined by
\begin{equation}
\langle v_1\twast v_2,\check g\rangle= ((g\twast v_1)\twast
v_2)(0),\quad g\in E.
 \label{3.2}
\end{equation}
If  $E$ has the structure of a topological algebra with respect to
product~\eqref{1.1}, then the algebras ${\mathcal C}_L(E)$ and
${\mathcal C}_R(E)$ are its extensions. Both of them are unital
with the Dirac $\delta$-function as their identity. Every
$\twast$-multiplier of $E$  by duality determines the
corresponding operations on the dual space $E'$, namely
\begin{equation}
\langle w\twast v_1,g\rangle= \langle w, \check v_1\twast
g\rangle,\,\, \langle v_2\twast w,g\rangle= \langle w, g\twast
\check v_2\rangle,\quad w\in E', v_1\in {\mathcal C}_L(E), v_2\in
{\mathcal C}_R(E),
 \label{3.3}
\end{equation}
which further  extends the initial operation on $E$.

If  $E$ is invariant under complex conjugation, then the
involution $g\to g^*$ in $E$ induces an involution of  $E'$. In
this case $(v\twast g)^*=g^*\twast v^*$  and  there is hence a
canonical antilinear isomorphism between  ${\mathcal C}_L(E)$ and
${\mathcal C}_R(E)$. We now consider the intersection
\begin{equation}
{\mathcal C}_\theta(E)={\mathcal C}_{\theta, L}(E)\cap{\mathcal
C}_{\theta, R}(E)
 \label{3.4}
\end{equation}
The natural topology of ${\mathcal C}_\theta(E)$ is the least
upper bound of the topologies induced by those of  ${\mathcal
C}_{\theta,L}(E)$ and ${\mathcal C}_{\theta,R}(E)$. If
$E\cap{\mathcal C}_\theta(E)$ is dense in ${\mathcal C}_\theta(E)$
in this topology and $E$ consists of bounded continuous integrable
functions, then ${\mathcal C}_\theta(E)$ has the canonical
structure of a unital involutive algebra  with respect to the
$\twast$-product, because  formulas~\eqref{3.1} and \eqref{3.2}
with $v_1,v_2\in {\mathcal C}_\theta(E)$ define the same
functional in this case. Indeed, using the Fubini theorem, we can
easily  verify that under these conditions and for $v_1,v_2\in E$,
the integrals representing  the right-hand sides of these formulas
coincide, and our assertion follows by passing to the limit.

\section{The main theorems}
\label{S4}
We must say a few words about the terminology used
below. Let $E$ be a complete nuclear locally convex space. If it
is a vector subspace of another locally convex space  $E_0$ and
the identical embedding $E\to E_0$  is continuous, we say that $E$
is a complete nuclear subspace of $E_0$. We need the following
auxiliary statement.

\begin{proposition}\labelP{3}
Let $E$ be a complete nuclear subspace of the Schwartz space
$S(\oR^d)$. Then $E\mathbin{\widetilde{\otimes}}E$ is a complete
nuclear subspace of $S(\oR^{2d})$.
\end{proposition}

\begin{proof}
We recall that the completed projective tensor product of nuclear
spaces is nuclear~\cite{Sch}. Because
$S(\oR^{2d})=S(\oR^d)\mathbin{\widetilde{\otimes}}S(\oR^d)$, there
is a natural continuous linear map $\iota\colon
E\mathbin{\otimes}E\to S(\oR^{2d})$, and we need only  show that
its extension $\tilde \iota$ by continuity to
$E\mathbin{\widetilde{\otimes}}E$ is injective.  If $E$ is
complete and nuclear, then by  the Grothendieck theorem
(Sec.~IV.9.4 in ~\cite{Sch}) the space
$E\mathbin{\widetilde{\otimes}}E$ is identified with the space
$\mathcal B_e(E'_\sigma, E'_\sigma)$ of separately continuous
bilinear forms on $E'_\sigma\times E'_\sigma$, where the
subscripts $\sigma$ and $e$  respectively mean that  $E'$ is
equipped with the weak topology and that $\mathcal B$ is equipped
with the topology of biequicontinuous convergence. The canonical
map $E\times E\to B_e(E'_\sigma, E'_\sigma)$ takes each pair of
functions $(f,g)$ to the bilinear form
\begin{equation}
(f\otimes g)(u, v)= \langle u,f\rangle\langle v, g\rangle,\qquad
u, v\in E'.
 \label{4.1}
 \end{equation}
We let $j'$ denote the transpose of the embedding $j\colon E\to
S(\oR^d)$. Applying the above-noted identification also to the
Schwartz space, we see that  $\tilde \iota$ maps the bilinear form
$\mathbf b\in\mathcal B_e(E'_\sigma, E'_\sigma)$ to the bilinear
form on $S'\times S'$ whose value on a pair of distributions
$u,v\in S'$ is $\mathbf b(j'(u),j'(v))$. Because $j$ is injective,
$j'(S')$ is dense in $E'_\sigma$ and  $\tilde \iota(\mathbf b)=0$
hence implies $\mathbf b=0$, which completes the proof.
\end{proof}

If the conditions of Proposition~\ref{P3} are satisfied, then we
use $E(\oR^{2d})$ and $E\mathbin{\widetilde{\otimes}}E$
equivalently to denote  the same space. A simple proof of the
kernel theorem generalizing the famous Schwartz theorem concerning
$S(\oR^d)$ to any complete nuclear barrelled space is given
in~\cite{S10JMP}. It develops a construction proposed by
Grothendieck (Theorem~13 in Chap.~2 in \cite{Grot}). In practice,
$E(\oR^{2d})$ is determined by the same restrictions as $E(\oR^d)$
but imposed on functions of the doubled number of variables, see
examples in Sec.~6.

It is well known that the (ordinary) convolution of any tempered
distribution with any test function in the Schwartz space  $S$ is
a multiplier of this space with respect to  pointwise
multiplication. The next theorem establishes an analogue of this
property for larger classes of generalized functions and shows
that it is preserved under the noncommutative deformation of
convolution.  For any locally convex space $E\subset S$, a
continuous function $\mu$ is called a pointwise multiplier of $E$
if  $\mu g\in E$ for all $g\in E$ and the map  $g\to \mu g$ is
continuous.\footnote{The continuity condition is automatically
satisfied if the closed graph theorem is applicable to $E$.} We
let  $M(E)$ denote the set of all such multipliers and give it the
topology induced by that of the operator algebra $\mathcal L(E)$.

\begin{theorem}\labelT{1}
Let $E$ be a complete nuclear barrelled subspace of $S(\oR^d)$ and
let $\theta(s, t)$ be a symplectic form on $\oR^d$. If
$e^{-\tfrac{i}{2}\theta}\in M(E\mathbin{\widetilde{\otimes}}E)$
and the involutive transformation $h(s,t)\to h(s,s-t)$ is a
continuous automorphism of
$E(\oR^{2d})=E\mathbin{\widetilde{\otimes}}E$, then $E$ is an
algebra under the $\twast_\theta$-product, this product is
separately continuous, the functions $v\twast_\theta g$ and
$g\twast_\theta v$ are well defined and belong to  $M(E)$ for any
$g\in E$ and  any $v\in E'$, and the maps $(g,v)\to v\twast_\theta
g$ and $(g,v)\to g\twast_\theta v$ from $E\times E'$ into $M(E)$
are also separately continuous.
\end{theorem}

\begin{proof}
We first show  that if  $h\in E(\oR^{2d})$ and $s$ is fixed, then
$h(s,t)$ regarded as a function of the variable $t$ belongs to
$E$. As in the proof of Proposition~\ref{P3}, we identify  $h$
with a bilinear separately continuous form $\mathbf h$ on
$E'_\sigma\times E'_\sigma$. We note that
\begin{equation}
h(s,t)= \mathbf h(\delta_s,\delta_t),
 \label{4.2}
\end{equation}
where the functional  $\delta_s\in E'$ is defined by $\langle
\delta_s, f\rangle=f(s)$. Indeed, \eqref{4.1} implies~\eqref{4.2}
for all linear combinations of functions of the form $h=f\otimes
g$ and~\eqref{4.2} holds by continuity for any element of
$E(\oR^{2d})$. The linear map $E'\to\oC\colon v\to \mathbf
h(\delta_s, v)$ is continuous in the weak topology $\sigma(E',E)$.
Therefore, there exists a function  $h_s\in E=(E'_\sigma)'$ such
that $\langle v, h_s\rangle=\mathbf h(\delta_s, v)$ for all $v\in
E'$. Substituting  $v=\delta_t$ here, we obtain $h_s(t)=h(s,t)$.
Clearly, the map $h\to h_s$ is continuous if  $E(\oR^{2d})$ and
$E$ are endowed with the weak topologies. Furthermore, it is easy
to see that the function $s\to \langle v,
h(s,\cdot)\rangle=\mathbf h(\delta_s, v)$ belongs to $E$. Indeed,
the linear map $u\to \mathbf h(u, v)$ is continuous in the
topology $\sigma(E',E)$ and  there  hence exists a function
$h_v\in E$ such that $\langle u, h_v\rangle=\mathbf h(u, v)$ for
all $u\in E'$. For $u=\delta_s$, this equality becomes
$h_v(s)=\langle v, h(s,\cdot)\rangle$. The map $E'\to E\colon v\to
h_v$ is also continuous in the weak topologies. Moreover, $h_v$
depends weakly continuously on $h$.

 We  apply the above consideration to the function $h(s,t)=f(s)g(s-t)$,
where $f,g\in E$. By  the conditions in the theorem, it belongs to
$E(\oR^{2d})$. We fix a point $s\in \oR^d$ and choose  $f$ such
that $f(s)=1$.  By what was said above, the function  $t\to
g(s-t)$ belongs to $E$ and the map $E\to E\colon g(t)\to g(s-t)$
is weakly continuous. Every barrelled space is a Makkey space, it
hence follows that this map is continuous (Sec.~IV.7.4 in
\cite{Sch}). Because $s$ can be chosen arbitrarily, we see that
the coordinate reflection and translations are continuous
operators on $E$. Similar reasoning applied to
$f(s)e^{-\tfrac{i}{2}\theta(s, t)}g(s-t)$ shows that the function
$t\to e^{-\tfrac{i}{2}\theta(s, t)}g(s-t)$ belongs to $E$ and the
map $g(t)\to e^{-\tfrac{i}{2}\theta(s, t)}g(s-t)$ is continuous.
Therefore,  twisted shifts~\eqref{2.1} and \eqref{2.2} are also
continuous operators on $E$,  and the definitions of $v\twast g$
and $g\twast v$  in Sec.~\ref{S2} are applicable here.
Furthermore, because $f(s)\langle v,e^{\pm\tfrac{i}{2}\theta
(s,\cdot)}g(s-\cdot)\rangle\in E$ for any $f\in E$, the functions
$\langle v,\tau_s \check g\rangle$ and $\langle v,\bar\tau_s
\check g\rangle$ are continuous in $s$ for each $v\in E'$. Hence,
the sets  $\{\tau_s \check g \colon |s|\le1\}$ and $\{\bar\tau_s
\check g \colon |s|\le1\}$, where  $g$ is fixed, are weakly
bounded. These sets are then bounded in the original topology of
$E$ (Sec.~IV.3.2 in  \cite{Sch}) and  the corresponding
representations of the Weyl-Heisenberg group in $E$ are continuous
by~\refL{1} proved in the appendix. For any fixed $v$ and $g$, the
functions $f\cdot (v\twast g)$ and $f\cdot (g\twast v)$ are weakly
continuous in $f$, and  the multiplication by $v\twast g$ and by
$g\twast v$ are therefore continuous operators from $E$ to $E$. We
conclude that the twisted convolution products belong to $M(E)$.
Furthermore, again by what was said above, the function $t\to
\langle u,f(\cdot)e^{-\tfrac{i}{2}\theta(\cdot,
t)}g(t-\cdot)\rangle$ belongs to $E$ for any $u\in E'$. In
particular,
\begin{equation}
\int f(s)e^{-\tfrac{i}{2}\theta(s, t)}g(t-s)\,ds= f\twast g\in E.
\notag
\end{equation}
 Therefore, $E$ is an algebra under the
twisted convolution. Clearly, the maps $f\to f\twast g$ and $g\to
f\twast g$ are  weakly continuous and  are hence continuous in the
original topology of $E$.

It remains to prove that the maps $(g,v)\to v\twast g$ and
$(g,v)\to g\twast v$ from $E\times E'$ to $M(E)$ are separately
continuous. Let $v$ be fixed. We show that for each neighborhood
$U$ of the origin and for each bounded set $Q$ in $E$, there
exists a neighborhood  $W$ of the origin  such that
$f\cdot(v\twast g)\in U$ for all $f\in Q$ and for all $g\in W$. By
the definition of the topology  of biequicontinuous convergence, a
base of neighborhoods of the origin in
$E\mathbin{\widetilde{\otimes}}E=\mathcal B_e(E'_\sigma,
E'_\sigma)$ is formed by sets of the form
\begin{equation}
\mathcal U_{\,U,V}=\left\{\mathbf b\in \mathcal B(E'_\sigma,
E'_\sigma)\colon \sup_{u\in U^\circ, v\in V^\circ}|\mathbf
b(u_,v)|\le 1\right\} ,
 \label{4.3}
\end{equation}
where  $U$ and $V$ run over a basis of absolutely convex closed
neighborhoods of zero in  $E$, and  $U^\circ$, $V^\circ$ are their
polars, i.e., $U^\circ=\{u\in E'\colon \sup_{f\in U}|\langle
u,f\rangle|\le 1\}$. The functional $v$ is bounded by unity on a
zero-neighborhood which we take as   $V$ in~\eqref{4.3}. By the
conditions of the theorem, for any  $\mathcal U_{\,U,V}$, we can
find neighborhoods $U_1$ and $V_1$ such that
\begin{equation}
h(s,t)=f(s)e^{\tfrac{i}{2}\theta(s, t)}g(s-t)\in \mathcal
U_{\,U,V} \notag
\end{equation}
 for all $f\in U_1$ and $g\in V_1$. For the bounded set
$Q$, there exists $\delta>0$ such that $Q\subset \delta U_1$. We
set $W=\delta V_1$. Then for all $u\in U^\circ$, $f\in Q$, and
$g\in W$, we have
\begin{equation}
|\mathbf h(u,v)|=|\langle u,f(v\twast g)\rangle|\le 1 .
 \notag
\end{equation}
In another words, $f\cdot(v\twast g)\in U^{\circ\circ}=U$, which
proves the continuity of the map $E\to M(E)\colon g\to v\twast g$.

Now let $g$ be fixed and a sequence $v_n$ tend to zero in $E'$.
Then the sequence of multipliers $v_n\twast g\to 0$ in the
topology of simple convergence on $L(E)$ and hence also in the
topology of bounded convergence (see Sec.~III.4.6 in~\cite{Sch})
because $E$ is a Montel space as noted in Remark~\ref{R2}. Hence,
the map $E'\to \mathcal L(E)\colon v\to v\twast g$ is sequentially
continuous. It remains to show that  $E'$ is a bornological space
because for such spaces the sequential continuity of a linear map
into an arbitrary locally convex space implies its continuity
(Sec.~II.8.3 in~\cite{Sch}). For this, we recall that every
complete nuclear space is representable as the projective limit of
a suitable family of Hilbert spaces $E_\gamma$. The projective
limit can be assumed to be reduced and the dual space  $E'$  with
the Mackey topology $\tau(E',E)$ is then the inductive limit of
the family of spaces $E'_\gamma$ equipped with the Makkey
topologies $\tau(E'_\gamma,E_\gamma)$ (Sec.~IV.4.4 in~\cite{Sch}).
Because the Montel and Hilbert spaces are reflexive, the strong
topology of their duals coincides with the Makkey topology.
Therefore, $E'$ is bornological as an  inductive limit of normed
spaces. We conclude the the map $E'\to M(E)\colon v\to v\twast g$
is continuous. A similar argument for $g\twast v$ completes the
proof.
\end{proof}

\begin{corollary}\labelC{1} If the conditions of Theorem~\ref{T1} are
satisfied, then definition~\eqref{2.3} is equivalent to the
definition of the products  $v\twast g$ and $g\twast v$ by
duality, and the relations
\begin{equation}
(v\twast g)\twast f=v\twast(g\twast f),\quad (g\twast v)\twast f=
g\twast (v\twast f),\quad g\twast(f\twast v)=(g\twast f)\twast v.
 \label{4.4}
\end{equation}
hold for all $v\in E'$, and  $g,f\in E$.
\end{corollary}

\begin{proof}
For any topological algebra $(E,\twast)$, the twisted convolution
of a function $g\in E$ with a functional $v\in E'$ is defined by
duality by the formulas
\begin{equation}
\langle v\twast g,f\rangle= \langle v, \check g\twast
f\rangle,\,\, \langle g\twast v,f\rangle= \langle v, f\twast
\check g\rangle,\quad f\in E.
 \label{4.5}
\end{equation}
The maps $v\to v\twast g$ and $v\to g\twast v$ from $E'$ to $E'$
are then continuous because they are the transposes  of the
continuous maps $f\to\check g\twast f$ and $f\to f\twast\check g$.
For $v\in E$, the right-hand sides of the equalities
in~\eqref{4.5} are easily seen to coincide with the result of
integrating functions~\eqref{2.3} multiplied by the test function
$f$. We must show that this also holds for any $v\in E'$. For
this, we note that under the conditions in Theorem~\ref{T1}, the
space $E$ is dense in  $E'$. Indeed, as shown in the proof
of~\refT{1},  along with a function $g(t)$, $E$ contains all the
shifted functions $g(t-s)$ and also all the products
$e^{\tfrac{i}{2}\theta(s, t)}g(t)$, where $s\in \oR^d$. According
to Sec.~IV.8.4 in~\cite{GS2}, it follows that $E$ has sufficiently
many functions in the following sense. Let a function $\varphi$ be
locally integrable. Then the condition that the integral $\int
\varphi(t)f(t)\,dt$ exists and is zero for all $f\in E$ implies
that $\varphi(t)\equiv 0$. Therefore, the canonical map $E\to E'$
is injective, and since $E$ is reflexive, this map has a closed
range by the Hahn-Banach theorem, as stated above. Further, each
multiplier $m\in M(E)$ defines a functional $\mu\in E'$ by the
formula $\langle \mu,f\rangle=\int m(t)f(t)\,dt$, and the map
\begin{equation}
M(E)\to E'\colon m\to \mu.
 \label{4.6}
\end{equation}
is injective by the same reasoning and is continuous by the
definition of topology in these spaces. Approximating the
functional $v\in E'$ by functions $v_\gamma\in E$, using the
continuity of the maps $v\to v\twast g$, $v\to g\twast v$ from
$E'$ to $M(E)$, and passing to the limit in the equality
\begin{equation}
\int (v_\gamma\twast g)(t)f(t)dt=\langle v_\gamma, \check g\twast
f\rangle
 \notag
\end{equation}
 and in an analogous equality for $g\twast v_\gamma$, we
conclude that map~\eqref{4.6} establishes a one-to-one
correspondence between the functions defined by~\eqref{2.3} and
the functionals defined by~\eqref{4.5}. Formulas~\eqref{4.4}
follow by continuity from the associativity of the algebra $(E,
\twast)$, which completes the proof.
\end{proof}

We now note that, under the conditions of Theorem~\ref{T1}, the
space $E$ is an algebra under  pointwise multiplication. Indeed,
for any $f,g\in E$, the function $f(s)g(s-t)|_{t=0}$ belongs to
$E$, and it follows from the proof of Theorem~\ref{T1} that the
map $(f,g)\to f\cdot g$ is separately continuous. Hence, there is
a canonical embedding  $E\to M(E)$. We assume below that $E$ is
dense in $M(E)$. This condition is generally assumed for the
multipliers of topological algebras and  is included in the
definition of $M(E)$ in~\cite{P}. If the density condition is
satisfied, then the algebra $M(E)$ is canonically isomorphic to
the closure in $\mathcal L(E)$ of the set of all operators of
multiplication by elements of $E$.

\begin{theorem}\label{T2} If the conditions of Theorem~\ref{T1} are
satisfied and  $E$ is dense in $M(E)$, then the space $M'(E)$ dual
to $M(E)$ is contained in $\mathcal C_\theta(E)$.
\end{theorem}

\begin{proof}
Let  $Q$ be a bounded set in $E$. The set of all continuous maps
$g\to f\cdot g$, where $f\in Q$, is pointwise bounded. By the
Banach-Steinhaus theorem, it follows that for each neighborhood
$U$ of the origin in  $E$, there exists a neighborhood $V$ such
that $f\cdot g\in U$ for all $g\in V$ and  $f\in Q$. This implies
that the natural injection  $E\to M(E)$ is continuous. Its
transpose $M'(E)\to E'$ is hence well defined and continuous. The
assumption that $E$ is dense  in $M(E)$ implies that the latter
map is injective. Therefore, if $w\in M'(E)$, then $w\in E'$ and
$w\twast g\in M(E)$ by Theorem~\ref{T1}. We must show that
$w\twast g\in E$. By Corollary~\ref{C1} we have
\begin{equation}
(w\twast g)\twast f=w\twast(g\twast f) \quad \text{for all $g,f\in
E$}.
 \notag
 \end{equation}
The equality of these convolution products  at zero can be written
as
\begin{equation}
\langle w\twast g,\check{f}\rangle=\langle\check{w},g\twast
f\rangle.
 \label{4.7}
 \end{equation}
We let $L'_g$ denote the linear map from $M'(E)$ to $E$ that is
transpose of the continuous map $L_g\colon v\to g\twast v$. We
then have
\begin{equation}
\langle v ,L'_g\check{w}\rangle=\langle\check{w},g\twast v\rangle
\quad \text{for all $v\in E'$.}
 \label{4.8}
 \end{equation}
For $v=f$, the right-hand side of~\eqref{4.7} coincides with that
of~\eqref{4.8}, and the left-hand side of \eqref{4.8} is written
as $\int (L'_g\check{w})(t)f(t)dt$. Because $f$ is an  arbitrary
element of $E$ and this space has sufficiently many functions, we
infer that the function $(w\twast g)(t)$ coincides with the
function $(L'_g\check{w})(-t)$ belonging to $E$. It remains to
show that the map $E\to E\colon g\to w\twast g$ is continuous. By
Theorem~\ref{T1}, for any fixed $v\in E'$, the map $E\to
M(E)\colon g\to g\twast v$ is continuous,  and hence {\it a
fortiori} continuous in the weak topologies of  $E$ and $M(E)$.
Therefore, the scalar-valued function  $g\to \langle\check{w},L_g
v\rangle$ is continuous in the topology  $\sigma(E,E')$, and  the
map $E\to E\colon g\to L'_g\check{w}=w\twast g$ is hence weakly
continuous. Since $E$ is a Mackey space, this map is also
continuous in its original topology. We conclude that $w\in
{\mathcal C}_L(E)$. Analogously, we obtain  $w\in {\mathcal
C}_R(E)$. The theorem is proved.
\end{proof}

Theorem~\ref{T2} implies a corresponding result for the space  $F$
that is Fourier-conjugate to $E$, i.e., $E=\widehat F$. We let
$\hat f(s)\coloneq\int f(x)e^{-i x\cdot s}dx$ denote the Fourier
transform of a function $f$. For any functions $f_1$ and $f_2$ in
the Schwartz space $S$, we have the relation
\begin{equation}
\widehat{(f_1\star_\theta f_2)}=(2\pi)^{-d}\,\hat
f_1\mathbin{\twast_\theta}\hat f_2,
\label{4.9}
 \end{equation}
which can be taken as the definition of the Moyal product of
elements of  $S$. Therefore, if  $E$ satisfies the conditions of
Theorems~\ref{T1} and \ref{T2}, then $F$ is an algebra under the
$\star_\theta$-product with separately continuous multiplication.
This product can be uniquely extended by continuity to the case
where one of the factors belongs to the dual space, and if  $u\in
F'$ and $f\in F$, then
\begin{equation}
\widehat{(u\star_\theta f)}=(2\pi)^{-d}\,\hat
u\mathbin{\twast_\theta}\hat f,\quad  \widehat{(f\star_\theta
u)}=(2\pi)^{-d}\,\hat f\mathbin{\twast_\theta}\hat u. \notag
 \end{equation}
Furthermore, we have the  relations
\begin{equation}
\langle u\star_\theta g,f\rangle= \langle u,  g\star_\theta
f\rangle,\,\, \langle g\star_\theta u,f\rangle= \langle u,
f\star_\theta g\rangle,\qquad f,g\in F,\quad u\in F',
 \label{4.10}
\end{equation}
which correspond to  formulas~\eqref{4.5} and extend the Moyal
product to elements of  $F'$ by duality. Because $E$ is an algebra
under  ordinary multiplication, $F$ is also an algebra under
ordinary convolution. We let $C(F)$ denote the corresponding
algebra of (ordinary) convolution multipliers and  $\mathcal
M_{\theta, L}(F)$ and $\mathcal M_{\theta, R}(F)$ denote  the
algebras of left and right Moyal multipliers for $F$. Then
$\widehat{C}(F)= M(\widehat F)$, $\widehat{\mathcal M}_{\theta,
L}(F)=\mathcal C_{\theta, L}(\widehat F)$, and $\widehat{\mathcal
M}_{\theta, R}(F)=\mathcal C_{\theta, R}(\widehat F)$. We thus
obtain the following result.

\begin{corollary}\label{C2}
If $F$ is a function space whose Fourier transform $E=\widehat F$
satisfies the conditions of Theorems~\ref{T1} and \ref{T2}, then
the space $C'(F)$ is contained in ${\mathcal
M}_\theta(F)={\mathcal M}_{\theta, L}(F)\cap{\mathcal M}_{\theta,
R}(F)$.
\end{corollary}

\section{The case of  Fourier-invariant spaces }
\label{S5} The Fourier-invariant function  spaces with the
structure of an algebra under both the twisted convolution and the
Moyal products are  particularly interest. We recall that an
autohomeomorphism of a topological space is  a continuous
bijection  of this space onto itself, whose inverse is also
continuous.

\begin{theorem}\label{T3}  If a space $E$ satisfies the conditions
of Theorems~\ref{T1} and \ref{T2} and, in addition, the Fourier
transform and the linear changes of variables in $\oR^d$ are
autohomeomorphisms of $E$, then both the spaces  $M'(E)$ and
$C'(E)$ are contained in ${\mathcal C}_\theta(E)$ and in
${\mathcal M}_\theta(E)$.
\end{theorem}

\begin{proof}
We can assume that the symplectic form $\theta$ is determined by
an antisymmetric matrix  $\vartheta$ via the formula
$\theta(s,t)=s\cdot\vartheta t$, where  the dot denotes the usual
inner product on  $\oR^d$. Let $f_1,f_2\in E$. It follows
from~\eqref{1.1} and~\eqref{4.9} that
\begin{multline}
(f_1\star_\theta f_2)(x)= \frac{1}{(2\pi)^{2d}}\iiint \hat f_1(t)
f_2(y)\,e^{\tfrac{i}{2}s\cdot\vartheta t-i(s-t)\cdot y +is\cdot x}dt\,dy\,ds=\\
=\frac{1}{(2\pi)^{d}}\iint \hat f_1(t) f_2(y)\,e^{it\cdot
y}\delta\left(y-x-\tfrac12\vartheta t\right)  dt \,
dy=\frac{1}{(2\pi)^{d}}\int  \hat f_1(t) f_2(x+\tfrac12\theta
t)\,e^{it\cdot x} dt=
\\=\frac{1}{\pi^d\det\theta}\int  \hat f_1(-2\theta^{-1}\xi)
f_2(x-\xi)\,e^{-2ix\cdot\theta^{-1}\xi}d\xi.
 \label{5.1}
\end{multline}
Hence, we have
\begin{equation}
f_1\star_\vartheta f_2= \frac{1}{\pi^d\det\theta} (\mathcal
F_\vartheta f_1)\twast_{-4\vartheta^{-1}} f_2,
 \notag
\end{equation}
where
\begin{equation}
(\mathcal F_\vartheta f)(\xi)\coloneq \int
f(x)e^{2ix\cdot\vartheta^{-1}\xi}dx.\notag
\end{equation}
An analogous calculation gives
\begin{equation}
f_1\star_\vartheta f_2= \frac{1}{\pi^d\det\theta}
f_1\twast_{-4\vartheta^{-1}} (\overline{\mathcal F}_\vartheta
f_2), \notag
\end{equation}
where
\begin{equation}
(\overline{\mathcal F}_\vartheta f)(\xi)\coloneq \int
f(x)e^{-2ix\cdot\vartheta^{-1}\xi}dx.\notag
\end{equation}
It follows from the  conditions of the theorem that each of the
two symplectic Fourier transforms  $\mathcal F_\vartheta$ and
$\overline{\mathcal F_\vartheta}$ is an  autohomeomorphism of $E$,
and, furthermore, $e^{-2i x\cdot\vartheta^{-1}\xi}\in
M(E\mathbin{\widetilde{\otimes}}E)$. Therefore, by
Theorem~\ref{T1}, the twisted convolution products $(\mathcal
F_\vartheta f)\twast_{-4\vartheta^{-1}} v$ and
$v\twast_{-4\vartheta^{-1}}(\overline{\mathcal F}_\vartheta f)$
with the deformation parameter $-4\vartheta^{-1}$ are well defined
for all $f\in E$ and for all $v\in E'$, and moreover, are
continuous in $v$. Because the twisted convolution and the  Moyal
multiplication extend   by continuity uniquely to the case where
one of the factors belongs to the dual space, we conclude  that
\begin{equation}
f\star_\vartheta v= \frac{1}{\pi^d\det\theta}(\mathcal F_\theta
f)\twast_{-4\vartheta^{-1}} v,\quad v\star_\vartheta
f=\frac{1}{\pi^d\det\theta} v\twast_{-4\vartheta^{-1}}
(\overline{\mathcal F}_\theta f).
 \label{5.2}
\end{equation}
By Theorem~\ref{T1} and formulas~\eqref{5.2}, both the products
$f\star_\vartheta v$ and $v\star_\vartheta f$ belong to $M(E)$,
and the maps $(f,v)\to f\star_\vartheta v$ and $(f,v)\to
v\star_\vartheta f$ from $E\times E'$ to  $M(E)$ are continuous.
Further arguments are similar to those  in the proof of
Theorem~\ref{T2}. Let $w\in M'(E)\subset E'$. Then
$w\star_\vartheta f\in M(E)$ by Theorem~\ref{T1}. We need to show
that $w\star_\vartheta f\in E$. From~\eqref{4.10}, we have
\begin{equation}
\langle w\star_\vartheta f,g\rangle=\langle w,f\star_\vartheta
g\rangle \quad \text{for all $g\in E$}.
 \label{5.3}
 \end{equation}
Let $L_f$ be the linear map $v\to f\star_\theta v$ from $E'$ to
$M(E)$. Then
\begin{equation}
\langle  v, L'_f w\rangle=\langle w,f\star_\vartheta v\rangle
\quad \text{for all $v\in E'$}.
 \label{5.4}
 \end{equation}
For $v=g$, the right-hand sides of equalities~\eqref{5.3} and
\eqref{5.4} coincide and
 \begin{equation}
 \langle g, L'_f w\rangle=\int  (L'_f
w)(\xi)g(\xi)d\xi.
 \notag
 \end{equation}
Therefore,  $w\star_\vartheta f$ coincides with the function $L'_f
w$ belonging to $E$. For any fixed $v\in E'$, the map $E\to
M(E)\colon f\to f\star_\vartheta v$ is continuous, and {\it a
fortiori} continuous in the weak topologies of $E$ and $M(E)$. The
function $f\to \langle w,L_f v\rangle$ is hence continuous in the
topology $\sigma(E,E')$. Consequently, the map $E\to E\colon f\to
L'_f w=w\star_\vartheta f$ is weakly continuous. Because $E$ is a
Mackey space, this map is also continuous in its original
topology. We conclude that $w\in {\mathcal M}_{\theta,L}(E)$.
Analogously, $w\in {\mathcal M}_{\theta,R}(E)$, and hence
$M'(E)\subset {\mathcal M}_\theta(E)$. From this result combined
with the isomorphisms $\widehat{M}(E)=C(\widehat E)$ and
$\widehat{{\mathcal M}}_\theta(E)={\mathcal C}_\theta(\widehat
E)$, where in this case $\widehat E =E$, we deduce that
$C'(E)\subset\mathcal C_\theta(E)$. The theorem is proved.
\end{proof}

\begin{remark}\label{R3} It follows from formulas~\eqref{5.2} that under
the conditions of Theorem~\ref{T3}, the algebras ${\mathcal
M}_\vartheta$ and ${\mathcal
C}_{-4\vartheta^{-1}}(E)=\widehat{{\mathcal
M}}_{-4\vartheta^{-1}}(E)$ consist of the same elements of $E'$.
In particular, for $\theta=\hbar J$ with $J$ being the standard
symplectic matrix, the algebras ${\mathcal M}_{\hbar J}$ and
${\mathcal C}_{4\hbar^{-1}J}(E)$ consist of the same elements, and
the algebra ${\mathcal M}_{2J}(E)$ is Fourier invariant.
\end{remark}

\section{Examples and concluding remarks}
\label{S6}
 Many of the spaces used in the theory of generalized
functions are nuclear Fr\'echet spaces or their strong duals. They
are  respectively abbreviated as  FN  and DFN spaces. These spaces
have additional nice properties. In particular, the Ptak version
of the closed graph theorem is applicable to them and the separate
continuity of a bilinear map of the Cartesian product of such
spaces to an arbitrary locally convex space is equivalent to the
joint continuity of this map. Furthermore, the completed
projective tensor product of two FN or DFN spaces is respectively
an FN space or a DFN space (see.~\cite{Sch}, \cite{K}). We note
that every FN space is also an FS  (Fr\'echet-Schwartz) space and
every DFN space is a DFS space.  A survey of basic properties of
FS and DFS spaces can be found, for example, in~\cite{Zh},
\cite{Mor}. The formulations and proofs given above can be
considerably simplified for these spaces. The Schwartz space $S$
is the best-known example of an FN space. The spaces
$\mathscr{S}^\beta$ considered in~\cite{S07}, \cite{S10} in the
context of noncommutative field theory also belong to this class.
The space $\mathscr{S}^\beta(\oR^d)$, $\beta>0$, is defined as the
projective limit
$\projlim\limits_{N\to\infty,B\to0}S^{\beta,B}_N(\oR^d)$, where
$S_N^{\beta,B}(\oR^d)$ is the Banach space of  smooth functions on
$\oR^d$ with the finite norm
\begin{equation}
\|f \|_{N,B}=\sup_{x,\kappa}\,(1+|x|)^N\frac{|\partial^\kappa
f(x)|}{B^{|\kappa|}\kappa^{\beta\kappa}}.
 \label{6.1}
\end{equation}
(Here and hereafter, we use the standard   multi-index notation
adopted in~\cite{GS2}, \cite{H1}.) The Fourier transformed space
$\widehat{\mathscr{S}}^\beta=\mathscr{S}_\beta$ belongs to the
class of spaces $K(M_n)$, where in this case
$M_n(p)=e^{n|p|^{1/\beta}}$, and satisfies the nuclearity
condition $\int (M_n/M_{n'})dp<\infty$, $n'>n$, indicated
in~\cite{GS2}. If $\beta>1$, then the space $\mathscr{S}^\beta$
contains functions of compact support, and the elements of its
dual $(\mathscr{S}^\beta)'$ are said to be tempered
ultradistributions (of the Beurling type). For $\beta\le 1$, the
functions in $\mathscr{S}^\beta(\oR^d)$ allow analytic
continuation  to $\oC^d$ as entire functions of order $\le
1/(1-\beta)$. The space $\mathscr{S}^1$ plays a particular role.
Following Morimoto~\cite{Mor2}, we call the elements of
$(\mathscr{S}^1)'$  tempered ultrahyperfunctions; they were used
in axiomatic formulation of nonlocal quantum field theory with a
fundamental length (see~\cite{S10JMP}, \cite{BN}, \cite{S09JMP}
and references therein). It is easy to see that
$e^{-\tfrac{i}{2}\theta}$ is a pointwise multiplier of
$\mathscr{S}_\beta(\oR^{2d})$  for any real bilinear form $\theta$
on $\oR^d$. The space $\mathscr{S}_\beta$ thus satisfies all the
conditions of Theorem~\ref{T1} and is hence an algebra under the
star product $\star_\theta$ for any $\beta>0$. In the study of
noncommutative quantum field theory models, it is common practice
to represent the Moyal product as the power series
\begin{equation}
(f\star_\theta g)(x)
=f(x)g(x)+\sum_{n=1}^\infty\left(\frac{i}{2}\right)^n\frac{1}{n!}\,
\vartheta^{\mu_1\nu_1}\dots
\vartheta^{\mu_n\nu_n}\partial_{\mu_1}\dots
\partial_{\mu_n}f(x)\partial_{\nu_1}\dots\partial_{\nu_n}g(x),
 \label{6.2}
\end{equation}
where $\vartheta^{\mu\nu}$ is the matrix of the symplectic form
$\theta$. It was shown in~\cite{S07} that if $\beta\le1/2$, then
series~\eqref{6.2} converges absolutely for  all
$f,g\in\mathscr{S}^\beta$ in each of  norms~\eqref{6.1}. The space
$\mathscr{S}^{1/2}$ is  a maximal FN subspace of $S$ for which
representation~\eqref{6.2} is absolutely convergent, and this
space hence  plays a  special role in noncommutative field theory.
Let $E^{\beta,B}_N(\oR^d)$ be the Banach space of smooth functions
on $\oR^d$ with the finite norm
\begin{equation}
\|f \|_{-N,B}=\sup_{x,\kappa}\,(1+|x|)^{-N}\frac{|\partial^\kappa
f(x)|}{B^{|\kappa|}\kappa^{\beta\kappa}}
 \label{6.3}
\end{equation}
and let $\mathscr{E}^\beta=\injlim\limits_{N\to\infty}
\projlim\limits_{B\to0}E^{\beta,B}_N$.

\begin{theorem}\label{T4}
The space $\mathscr{E}^\beta$ is contained in $\mathcal
M_\theta(\mathscr{S}^\beta)$. For each $w\in \mathscr{E}^\beta$
and  each $v\in (\mathscr{S}^\beta)'$, the Moyal products
$w\star_\theta v$ and $v\star_\theta w$ are well defined as
elements of $(\mathscr{S}^\beta)'$.
\end{theorem}

\begin{proof}
By Corollary~\ref{C2}, to  prove the first statement, it suffices
to show the inclusion $\mathscr{E}^\beta\subset
C'(\mathscr{S}^\beta)$. Then the second statement  also holds
because the products under consideration are defined by duality
for any $w\in \mathcal M_\theta(\mathscr{S}^\beta)$ and $v\in
(\mathscr{S}^\beta)'$. As proved in~\cite{P}, the algebra
$C(\mathscr{S}^\beta)$ consists of the same elements as the
space\footnote{In~\cite{P}, the notation $S^{\beta-}$ and
$\hat{\mathscr{E}}^{\beta-}$ was used instead of
$\mathscr{S}^\beta$ and $\mathscr{E}^\beta$. The space of
pointwise multipliers for $\mathscr{S}^\beta$ is shown there to be
$\projlim\limits_{B\to0}\injlim\limits_{N\to\infty}
E^{\beta,B}_N$.}  $(\mathscr{E}^\beta)'$. We want to show that the
topology of $C(\mathscr{S}^\beta)$ is not weaker than
$\sigma((\mathscr{E}^\beta)', \mathscr{E}^\beta)$ and hence
$C'(\mathscr{S}^\beta)\supset
((\mathscr{E}^\beta)'_\sigma)'=\mathscr{E}^\beta$.

If a functional $u\in (\mathscr{S}^\beta)'$ belongs to
$C(\mathscr{S}^\beta)$, then its continuous extension to
$\mathscr{E}^\beta$ can be constructed as follows. Let $f_0$ be a
function in $\mathscr{S}^\beta$ with the property that $\int
f_0(\xi)d\xi=1$, and let $h\in \mathscr{E}^\beta$. We set
\begin{equation}
\langle  u,h\rangle=\int \langle u,h(\cdot) f_0(\xi-\cdot)\rangle
d\xi=\int(u*h_\xi)(\xi) d\xi,\quad \text{where $h_\xi(x)\defeq
h(\xi-x)f_0(x)$}.
 \label{6.4}
\end{equation}
(the asterisk $*$ here  denotes ordinary convolution). The
integrand in~\eqref{6.4} is a continuous function  because
translations act continuously on $\mathscr{S}^\beta$ and $h$ is a
pointwise multiplier of this space. There exists an $N_0$ such
that $\|h\|_{-N_0,B}<\infty$ for each $B>0$, and for the function
$h_\xi\in \mathscr{S}^\beta$, we  hence have the estimate
\begin{multline}
|\partial^\kappa h_\xi(x)|\le \sum_\mu\binom{\kappa}\mu
|\partial^\mu h(\xi-x)\partial^{\kappa-\mu} f_0(x)|\\\le
\|h\|_{-N_0,B}\|f_0\|_{N,B}\sum_\mu\binom{\kappa}\mu
B^{|\mu|}\mu^{\beta\mu}
B^{|\kappa-\mu|}(\kappa-\mu)^{\beta(\kappa-\mu)}
(1+|\xi-x|)^{N_0}(1+|x|)^{-N}\\\le \|h\|_{-N_0,B}\|f_0\|_{N,B}
(2B)^{|\kappa|}\kappa^{\beta\kappa}(1+|x|)^{N_0-N}(1+|\xi|)^{N_0}.
 \label{6.5}
\end{multline}
Therefore, we find that for any $N,B>0$,
\begin{equation}
\|h_\xi\|_{N,B}\le
\|h\|_{-N_0,B/2}\|f_0\|_{N+N_0,B/2}(1+|\xi|)^{N_0}. \label{6.6}
\end{equation}
Because $u\in C(\mathscr{S}^\beta)$, for the neighborhood
$U=\{g\colon\|g\|_{N_0+d+1,1}<1\}$ in  $\mathscr{S}^\beta$, there
exists a neighborhood $V=\{g\colon\|g\|_{N,B}\le\delta\}$ such
that $u*g\in U$ for all $g\in V$. We let $\lambda_\xi$ denote the
right-hand side of inequality~\eqref{6.6}. Then
$\delta\lambda_\xi^{-1}h_\xi\in V$, and consequently
\begin{equation}
|(u* h_\xi)(\xi)|\le \delta^{-1}\lambda_\xi(1+|\xi|)^{-N_0-d-1}=
\delta^{-1}\|h\|_{-N_0,B/2}\|f_0\|_{N+N_0,B/2}(1+|\xi|)^{-d-1}.
 \label{6.7}
\end{equation}
Thus, the integral in~\eqref{6.4} converges absolutely and the
functional $u$ is well defined on $\mathscr{E}^\beta$. Since the
right-hand side of~\eqref{6.7} contains the factor
$\|h\|_{-N_0,B/2}$, this functional is continuous on every space
$\projlim\limits_{B\to0}\mathscr{E}^{\beta,B}_N$, and  is hence
continuous on $\mathscr{E}^\beta$. We must also show that the
functional defined by~\eqref{6.4}  when restricted to
$\mathscr{S}^\beta$ coincides with the original one, which
justifies using the same symbol  for them. If $h\in
\mathscr{S}^\beta$, then an estimate similar to~\eqref{6.5} yields
the inequality
\begin{equation}
\|h(x)f_0(\xi-x)\|_{N,B}\le
\|h\|_{N+N_0,B/2}\|f_0\|_{N_0,B/2}(1+|\xi|)^{-N_0}, \label{6.8}
\end{equation}
which holds for any $N_0$,$N$, and $B$. It follows that the
integral in~\eqref{6.4} remains absolutely convergent when $u$ is
replaced with any functional $v\in(\mathscr{S}^\beta)'$, because
$\|v\|_{N,B}\le \infty$ for some $N$ and $B$. The corresponding
sequence of integral sums is hence weakly fundamental in
$\mathscr{S}^\beta$. But $\mathscr{S}^\beta$, being a Montel
space, is complete with respect to  the weak convergence and this
convergence implies the strong convergence (see Sec.~I.6.3
in~\cite{GS2}). Therefore, the sequence of integral sums converges
in $\mathscr{S}^\beta$, and its limit  inevitably equals $h$,
because the topology of $\mathscr{S}^\beta$ is stronger than the
topology of simple convergence.

Now let $u_\gamma\to 0$ in $C(\mathscr{S}^\beta)$. Then for the
bounded set
\begin{equation}
Q=\bigcap\limits_{N,B}\{g\colon \|g\|_{N,B}\le
\|h\|_{-N_0,B/2}\|f_0\|_{N+N_0,B/2}\}
\notag
\end{equation}
and for the neighborhood
$\epsilon\, U=\{g\colon\|g\|_{N_0+d+1,1}<\epsilon\}$ with
arbitrarily small $\epsilon$, we can find $\gamma_0$ such that
$u_\gamma*g\in \epsilon\, U$ for all $g\in Q$ and all
$\gamma>\gamma_0$. By~\eqref{6.6} the family of functions
$(1+|\xi|)^{-N_0}h_\xi$ is  in $Q$ and we conclude that
\begin{equation}
|\langle u_\gamma,h\rangle|\le\int|(u_\gamma*h_\xi)(\xi)|d\xi\le
\epsilon\int(1+|\xi|)^{-d-1}d\xi\quad \text{for all
$\gamma>\gamma_0$.}
 \label{6.9}
\end{equation}
Therefore, $\langle u_\gamma,h\rangle\to 0$ for any $h\in
\mathscr{E}^\beta$, which completes the proof.
\end{proof}

The Gel'fand-Shilov spaces $S^\beta_\alpha$ are DFN spaces, and
the algebras $\mathcal M_\theta(S^\beta_\alpha)$ were constructed
and studied in~\cite{S11}. A special role of the minimal
nontrivial Fourier-invariant space $S^{1/2}_{1/2}$ was also shown
there. The other often used Gel'fand-Shilov spaces
$S^\beta=\injlim\limits_{B\to\infty}
\projlim\limits_{N\to\infty}S^{\beta,B}_N$ are neither Fr\'echet
nor DF spaces. This scale of spaces contains  $S^0$, which is the
Fourier transform of the space   $\mathscr D$ of smooth functions
with compact support. The possibility of using $S^0$ in
noncommutative field theory was discussed in~\cite{S06},
\cite{FP}. We note that the spaces $S^\beta$ with $\beta>0$ were
not given a topology in~\cite{GS2}, where the notion of
convergence of sequences  was used instead.  The above
consideration nevertheless applies to them because  these spaces
when equipped with a natural topology are nuclear~\cite{M},
complete~\cite{S01} and are obviously barrelled, being inductive
limits of Fr\'echet spaces. As shown in~\cite{P}, the algebra
$C(S^\beta)$ consists of the same elements as the dual of
$E^\beta=\injlim\limits_{N,B\to\infty}E^{\beta,B}_N=M(S^\beta)$.

\begin{theorem}\label{T5}
The space $E^\beta$ is contained in $\mathcal M_\theta(S^\beta)$.
For each $w\in E^\beta$ and  each $v\in (S^\beta)'$, the Moyal
products $w\star_\theta v$ and $v\star_\theta w$ are well defined
as elements of  $(S^\beta)'$.
\end{theorem}

\begin{proof}
Similarly to the above case of  $\mathscr{S}^\beta$, it suffices
to show that the topology of $C(S^\beta)$ is not weaker than
$\sigma((E^\beta)', E^\beta)$. A continuous extension of
$u\in(S^\beta)'$ to $E^\beta$ can be defined by the same
formula~\eqref{6.4}, but with $f_0\in S^\beta$. This time
$\|h\|_{-N_0,B_0}<\infty$ for some  $N_0,B_0>0$ and an estimate
analogous to~\eqref{6.5} gives
\begin{equation}
\|h_\xi\|_{N,B_0}\le
\|h\|_{-N_0,B_0/2}\|f_0\|_{N+N_0,B_0/2}(1+|\xi|)^{N_0}\quad
\text{for any $N>0$}. \label{6.10}
\end{equation}
Therefore, the family of functions $(1+|\xi|)^{-N_0}h_\xi$ is
bounded in $S^\beta$. Let $U$  be the neighborhood of the origin
in $S^\beta$ defined as the absolute convex hull of the set
\begin{equation}
\bigcup\limits_B\{g\colon \|g\|_{N_0+d+1,B}<1\}.
 \notag
\end{equation}
All functions in $U$ are then dominated by $(1+|x|)^{-N_0-d-1}$.
Because the image of any bounded set under the continuous map
$g\to u*g$ is absorbed by the neighborhood $U$, we see that the
integral in~\eqref{6.4} converges and determines a continuous
linear functional on $E^\beta$. If $h\in S^\beta$, then in this
case,~\eqref{6.8} holds for some $B$ and for any $N$ and $N_0$.
Hence, the integral in~\eqref{6.4} remains absolutely convergent
when $u$ is replaced with any linear functional defined and
continuous on the Montel space
$\projlim\limits_{N\to\infty,\varepsilon\to0}S^{\beta,B_1+\varepsilon}_N$,
where $B_1>B$, and we deduce that the constructed functional
restricted to $S^\beta$ coincides with the initial one. Finally,
let $u_\gamma\to 0$ in $C(S^\beta)$. Then for the bounded set
\begin{equation}
Q=\bigcap\limits_{N}\{g\colon \|g\|_{N,B_0}\le
\|h\|_{-N_0,B_0/2}\|f_0\|_{N+N_0,B_0/2}\}
 \notag
\end{equation}
and  the neighborhood $\epsilon\, U$, there exists  $\gamma_0$
such that $u_\gamma*g\in \epsilon\,U$ for all $g\in Q$ and  all
$\gamma>\gamma_0$. By~\eqref{6.10}, the family
$(1+|\xi|)^{-N_0}h_\xi$ is contained in $Q$, and we arrive
at~\eqref{6.9}, which completes the proof.
\end{proof}

We note that for any $\theta$, the algebras $\mathcal
M_\theta(\mathscr{S}^\beta)$ and $\mathcal M_\theta(S^\beta)$
contain  all polynomials, as does the Moyal multiplier algebra of
the Schwartz space. Furthermore, for any $\beta\ge\beta_0$, where
$0<\beta_0<1$, these algebras contain the entire  functions of
order $\le 1/(1-\beta_0)$ (and minimal type in the case of
$\mathcal M_\theta(\mathscr{S}^{\beta_0})$), that are polynomially
bounded on the real space. In particular, for any $\beta\ge0$, the
algebra $\mathcal M_\theta(S^\beta)$ contains the space
$\mathscr{O}_{\rm exp}$  consisting of  entire functions of
exponential type satisfying the condition $|f(z)|\le
C(1+|z|)^Ne^{b|\Im z|}$, where the constants $C$, $N$ and $b$
depend of $f$. This generalizes an earlier result~\cite{G-BV},
\cite{E} for $\mathcal M_\theta(S)$.

\begin{remark}\label{R4} Using reasoning similar to that for
$S^\beta_\alpha$ in the proof of Lemma~1 in~\cite{S11}, we can
show that $\mathscr{S}^\beta$ is dense in $\mathcal
M_\theta(\mathscr{S}^\beta)$. In accordance with what was said at
the end of~\refS{3}, it follows that $\mathcal
M_\theta(\mathscr{S}^\beta)$ is an algebra under the Moyal product
and $(\mathscr{S}^\beta)'$ is an $\mathcal
M_\theta(\mathscr{S}^\beta)$-bimodule. Analogously, $(S^\beta)'$
is a bimodule over the algebra $\mathcal M_\theta(S^\beta)$.
\end{remark}

The Weyl transforms of the constructed Moyal multiplier algebras,
or, in other words,  their operator realization in a Hilbert
space, will be considered in a forthcoming paper.

\section*{Appendix}

The following simple lemma is useful in  considering linear
representations of topological groups in  Montel function spaces.

\begin{lemma}\labelL{1}
Let $E$ be a Montel space of functions on a set $X$, and let the
topology of  $E$ is  stronger than that of pointwise convergence.
Let $\mathcal T$ be a linear representation of a locally compact
group  $\mathcal G$ in $E$. If for any function $f\in E$, there
exists a compact neighborhood $\mathcal K$ of unity $e$ in
$\mathcal G$ such that the set $\{\mathcal T_a f\colon
a\in\mathcal K\}$ is bounded in $E$ and if $\lim_{a\to e}
(\mathcal T_a f)(x)= f(x)$ for all $x\in X$, then the
representation $\mathcal T$ is continuous.
\end{lemma}

\begin{proof}
If  $\{a_n\}$ is a sequence converging to $e$ in $\mathcal G$,
then $a_n\in \mathcal K$ for sufficiently large  $n$, and the
sequence $\{\mathcal T_{a_n} f\}$ is hence bounded in $E$. By the
definition of a Montel space~\cite{Sch}, such a sequence has at
least one limit point. Since the topology of $E$ is stronger than
the topology of pointwise convergence, it follows from the limit
relation $\lim_{a\to e} (\mathcal T_a f)(x)= f(x)$ that only $f$
can be its limit point, and therefore $\mathcal T_{a_n} f\to f$.
Next, we use the generalized Banach-Steinhaus theorem (
Sec.~III.4.2 in~\cite{Sch}), which is applicable to  Montel
spaces. By that theorem, the pointwise boundedness of the operator
set $\{\mathcal T_a\}_{a\in\mathcal K}$ implies that for any
absolutely convex neighborhood $U$ of zero in $E$, there exists a
neighborhood  $V$ such that $\mathcal T_a(V)\subset U$ for all
$a\in \mathcal K$. Writing $\mathcal T_af -f_0= \mathcal
T_a(f-f_0)+\mathcal T_af_0-f_0$, we see that for the neighborhood
$U$ and  any function $f_0\in E$, there exists a neighborhood
$\mathcal K'(e)$ such that $\mathcal T_af\in f_0+ U$ for all $f\in
f_0+\frac12 V$ and  $a\in \mathcal K'$, i.e. the map $\mathcal
G\times E\to E\colon (a,f)\to\mathcal T_af$ is continuous at the
point $(e,f_0)$. Writing $\mathcal T_af-\mathcal
T_{a_0}f_0=\mathcal T_{a_0}(\mathcal T_{a_0^{-1}a}f-f_0)$, we
conclude that this map is continuous everywhere. The lemma is
proved.
\end{proof}

\begin{acknowledgements}  This paper was supported  by
the  Russian Foundation for Basic Research (Grant No.~09-01-00835)
\end{acknowledgements}


\begin{thebibliography}{99}

\bibitem{BS} F.~A.~Berezin and M.~A.~Shubin,
{\it Schr\"odinger Equation}, Kluwer, Dordrecht (1991).

\bibitem{ZFC} C.K.~Zachos, D.B.~Fairlie, and T.L~Curtright, eds.,
{\it Quantum Mechanics in Phase Space,} World Scientific,
Singapore (2005).

\bibitem{N} J. von Neumann, Die Eindentigkeit der
Schr\"odingerschen Operatoren, {\it Mathematische Annalen} {\bf
104}, 570  (1931).

\bibitem{A1} M.~A.~Antonets,
{\it The classical limit for Weyl quantization,}  Lett.\ Math.\
Phys.\  {\bf 2}, 241 (1978).

\bibitem{A2} M.~A.~Antonets, {\it The algebra of Weyl symbols and the
Cauchy problem for regular symbols,} Math. USSR-Sb. {\bf 35}, 317
(1979).

\bibitem{A3} M.~A.~Antonets, {\it Classical limit of Weyl
quantization,} Theor.\ Math.\ Phys.\  {\bf 38}, 219 (1979).

\bibitem{Mail} J.~M.~Maillard, {\it On the twisted convolution product
and the Weyl transformation of tempered distributions,}  J.\
Geom.\ Phys.\  {\bf 3}, 231 (1986).

\bibitem{G-BV}
J.~M.~Gracia-Bondia and J.~C.~V\'arilly, {\it Algebras of
distributions suitable for phase space quantum mechanics I,} J.\
Math.\ Phys.\  {\bf 29}, 869 (1988).

\bibitem{VG}
J.~C.~V\'arilly and J.~M.~Gracia-Bondia, {\it Algebras of
distributions suitable for phase-space quantum mechanics II.
Topologies on the Moyal algebra,}  J.\ Math.\ Phys.\   {\bf 29},
880 (1988).

\bibitem{E}
R.~Estrada, J.~M.~Gracia-Bondia, and J.~C.~V\'arilly, {\it  On
asymptotic expansions of twisted products,}  J.\ Math.\ Phys.\
{\bf 30}, 2789 (1989).

\bibitem{G}
V.~Gayral, J.~M.~Gracia-Bondia, B.~Iochum, T.~Sch\"ucker, and
J.~C.~V\'arilly, {\it Moyal planes are spectral triples,} Commun.\
Math.\ Phys.\  {\bf 246}, 569 (2004) [arXiv:hep-th/0307241].

\bibitem{Sz}
R.~J.~Szabo, {\it Quantum field theory on noncommutative spaces,}
Phys.\ Rept.\  {\bf 378}, 207 (2003) [arXiv:hep-th/0109162].

\bibitem{DFR1}
S.~Doplicher, K.~Fredenhagen, and J.~E.~Roberts, {\it Space-time
quantization induced by classical gravity,} Phys.\ Lett.\ B  {\bf
331}, 39 (1994).

\bibitem{DFR2}
 S.~Doplicher, K.~Fredenhagen, and J.~E.~Roberts,
{\it The Quantum structure of space-time at the Planck scale and
quantum fields,} Commun.\ Math.\ Phys.\  {\bf 172}, 187 (1995)
[arXiv:hep-th/0303037].

\bibitem{SW}
N.~Seiberg and E.~Witten, {\it String theory and noncommutative
geometry,} JHEP {\bf 9909}, 032 (1999) [arXiv:hep-th/9908142].

\bibitem{Alv}
L.~\'Alvarez-Gaum\'e and M.~A.~V\'azquez-Mozo, {\it General
properties of noncommutative field theories,}   Nucl.\ Phys.\ B
{\bf 668}, 293 (2003) [arXiv:hep-th/0305093].

\bibitem{HRR}
V.~E.~Hubeny, M.~Rangamani, and S.~F.~Ross, {\it Causal structures
and holography,} JHEP {\bf 0507}, 037 (2005)
[arXiv:hep-th/0504034].

\bibitem{S06}
M.~A.~Soloviev, {\it Axiomatic formulations of nonlocal and
noncommutative field theories,}  Theor.\ Math.\ Phys.\ {\bf 147},
660 (2006)
[arXiv:hep-th/0605249].

\bibitem{I}
M.~A.~Soloviev, {\it Star product algebras of test functions,}
Theor.\ Math.\ Phys.\  {\bf 153}, 1351 (2007)
[arXiv:0708.0811].

\bibitem{Ch}
M.~Chaichian, M.~Mnatsakanova, A.~Tureanu, and Y.~.Vernov, {\it
Test Functions Space in Noncommutative Quantum Field Theory,} JHEP
{\bf 0809}, 125 (2008)  [arXiv:0706.1712].

\bibitem{Green}
O.~W.~Greenberg, {\it Failure of microcausality in quantum field
theory on noncommutative spacetime,}   Phys.\ Rev.\ D  {\bf 73},
045014 (2006)  [arXiv:hep-th/0508057].

\bibitem{S07} M.~A.~Soloviev, {\it Noncommutativity and
$\theta$-locality,}  J.\ Phys.\ A \  {\bf 40}, 14593 (2007)
[arXiv:0708.1151].

\bibitem{S08}
M.~A.~Soloviev, {\it On the failure of microcausality in
noncommutative field theories,}  Phys.\ Rev.\  D {\bf 77}, 125013
(2008)   [arXiv:0802.0997].

\bibitem{GS2}  I.~M.~Gel'fand and G.~E.~Shilov, {\it Generalized
Functions. Vol.~2. Spaces of Fundamental and Generalized
Functions,} Academic Press, New York (1968).

\bibitem{P} V.~P.~Palamodov, {\it Fourier transforms of rapidly
increasing infinitely differentiable functions,} Trudy Moskov.
Mat. Obshch. {\bf 11}, 309 (1962)  [in Russian].

\bibitem{S11}
M.~A.~Soloviev, {\it Moyal multiplier algebras of the test
function spaces of type S,} J.\ Math.\ Phys.\  {\bf 52}, 063502
(2011)    [arXiv:1012.0669].

\bibitem{Sch} H.~H.~Schaefer, {\it Topological Vector Spaces,}
MacMillan, New York  (1966).

\bibitem{A} V.~I.~Arnold, {\it Mathematical Methods of Classical
Mechanics,} Springer-Verlag, New York (1989).

\bibitem{H1}   L.~H\"ormander, {\it The Analysis of Linear
Partial  Differential Operators I: Distribution Theory and Fourier
Analysis,} Springer-Verlag, Berlin (1983)

\bibitem{S10JMP}
M.~A.~Soloviev, {\it Reconstruction in quantum field theory with a
fundamental length,}  J.\ Math.\ Phys.\  {\bf 51}, 093520 (2010)
[arXiv:1012.3546].

\bibitem{Grot} A.~Grothendieck, {\it Produits tensoriels
topologiques et espaces nucl\'eaires}, Mem. Amer. Math. Soc.,
\textbf{16}, Amer. Math. Soc., Providence, RI (1955).


\bibitem{K}  G.~K\"othe, {\it Topological Vector Spaces II} (Grundlehren
Math. Wiss., Vol.~237), Springer-Verlag, New York (1979).

\bibitem{Zh} V.~V.~Zharinov, {\it Compact families of locally convex topological vector spaces
Fr\'echet-Schwartz and dual Fr\'echet-Schwartz spaces,} Russ.\
Math.\ Surveys\ {\bf 34}, No.~4, 105 (1979).

\bibitem{Mor} M.~Morimoto, {\it An Introduction to Sato's Hyperfunctions}
(Trasl. Math. Monogr., Vol.~129), Amer. Math. Soc., Providence, RI
(1993).

\bibitem{S10}
M.~A.~Soloviev, {\it Noncommutative deformations of quantum field
theories, locality, and causality,} Theor.\ Math.\ Phys.\ {\bf
163}, 741 (2010)
[arXiv:1012.3536].


\bibitem{Mor2} M.~Morimoto, {\it Theory of tempered ultrahyperfunctions I,}
Proc. Japan Acad. {\bf 51}, 87 (1975).


\bibitem{BN}  E.~Bruning and S.~Nagamachi,
{\it Relativistic quantum field theory with a fundamental length,}
J.\ Math.\ Phys.\  {\bf 45}, 2199 (2004).

\bibitem{S09JMP}  M.~A.~Soloviev,
{\it Quantum field theory with a fundamental length. A general
mathematical framework,}   J.\ Math.\ Phys.\  {\bf 50}, 123519
(2009) [arXiv:0912.0595].

\bibitem{FP} D.~H.~T.~Franco and C.~M.~M.~Polito,
{\it A new derivation of the CPT and spin-statistics theorems in
non-commutative field theories,} J.\ Math.\ Phys.\ {\bf 46},
083503 (2005) [arXiv:hep-th/0403028].

\bibitem{M} B.~S.~Mityagin, {\it Nuclearity and other properties of spaces of type
S}, Amer. Math. Soc. Transl., Ser. 2, {\bf 93}, Amer. Math. Soc.,
Providence, RI (1970), pp.~45-59.

\bibitem{S01}
M.~A.~Soloviev, {\it Lorentz-covariant ultradistributions,
hyperfunctions, and analytic functionals,} Theor.\ Math.\ Phys.\
{\bf 128}, 1252 (2001) [arXiv:math-ph/0112052].



\end{thebibliography}
\end{document}